\documentclass[journal]{IEEEtran}

\usepackage{xcolor,soul,framed} 

\usepackage[pdftex]{graphicx}
\graphicspath{{../pdf/}{../jpeg/}}
\DeclareGraphicsExtensions{.pdf,.jpeg,.png}

\usepackage{amsthm}
\usepackage{array}
\usepackage{mdwmath}
\usepackage{mdwtab}
\usepackage{eqparbox}
\usepackage{url}
\usepackage{empheq}
\usepackage{optidef}
\usepackage{tcolorbox}

 \renewcommand{\arraystretch}{1.5}
 \usepackage{makecell}

\renewcommand{\arraystretch}{1.0}
\newtheorem{thm}{Theorem}
\usepackage[framemethod=tikz]{mdframed}
\usepackage{algpseudocode}
\usepackage{algorithm}
\usepackage{multirow}
\usepackage{multicol}
\usepackage{lipsum}

\usepackage{booktabs}
\numberwithin{equation}{section}
\usepackage[font=footnotesize]{caption}
\usepackage{cite}
\usepackage{multirow}
\usepackage{dblfloatfix}
\usepackage{tabularx}
\usepackage{comment}
\usepackage{float}
\usepackage{url}
\usepackage{etoolbox}
\renewcommand\IEEEkeywordsname{Index Terms}

\usepackage{tikz,xcolor,hyperref}

\definecolor{lime}{HTML}{A6CE39}
\DeclareRobustCommand{\orcidicon}{%
	\begin{tikzpicture}
	\draw[lime, fill=lime] (0,0) 
	circle [radius=0.16] 
	node[white] {{\fontfamily{qag}\selectfont \tiny ID}};
	\draw[white, fill=white] (-0.0625,0.095) 
	circle [radius=0.007];
	\end{tikzpicture}
	\hspace{-2mm}
}
\foreach \x in {A, ..., Z}{%
	\expandafter\xdef\csname orcid\x\endcsname{\noexpand\href{https://orcid.org/\csname orcidauthor\x\endcsname}{\noexpand\orcidicon}}
}

\begin{document}
\bstctlcite{IEEEexample:BSTcontrol}
    \title{A Detection Mechanism against Load-Redistribution Attacks in Smart Grids}
\author{
Ramin~Kaviani\orcidA{}, \IEEEmembership{Student Member,~IEEE,}
Kory W.~Hedman\orcidB{},~\IEEEmembership{Senior Member,~IEEE}
\vspace{-0.5em}
  
\thanks
{
This work has been implemented to fulfill a part of the project: ``A Verifiable Framework for Cyber-Physical Attacks and Countermeasures in a Resilient Electric Power Grid'' funded by the National Science Foundation (NSF) Award under Grant 1449080. 

R. Kaviani and K. W.~Hedman are with the School of Electrical, Computer, and Energy Engineering, Arizona State University, Tempe, AZ 85281 USA  (e-mail: rkaviani@asu.edu;  kwh@myuw.net).
}
}
\markboth{IEEE TRANSACTIONS ON Smart Grid, VOL., NO., 2020
}{Ramin \MakeLowercase{\textit{et al.}}: Modified Energy Management System Including Load-Shifting Attack Detection and Enhanced SCED for Post-Attack Corrective Action}

\maketitle

\begin{abstract}
This paper presents a real-time non-probabilistic approach to detect load-redistribution (LR) attacks, which attempt to cause an overflow, in smart grids. Prior studies have shown that certain LR attacks can bypass traditional bad data detectors and remain undetectable, which implies that the presence of a reliable and intelligent detection mechanism is imperative. Therefore, in this study a detection mechanism is proposed based on the fundamental knowledge of the physics laws in electric grids. To do so, we leverage power systems domain insight to identify an underlying exploitable structure for the core problem of LR attacks, which enables the prediction of the attackers' behavior. Then, a fast greedy algorithm is presented to find the best attack vector and identify the most sensitive buses for critical transmission assets. Finally, a security index, which can be used in practice with minimal disruptions, is developed for each critical asset with respect to the identified best attack vector and sensitive buses. The proposed approach is applied to 2383-bus Polish test system to demonstrate the scalability and efficiency of the proposed algorithm.

\begin{IEEEkeywords}
cyber-attack detection, false data injection attack (FDIA), greedy algorithm, linear programming (LP), load-redistribution attack detection 
\end{IEEEkeywords}
\end{abstract}

\section*{Nomenclature}
\noindent
	\addcontentsline{toc}{section}{Nomenclature}
	\emph{\textbf{Sets and Indices}}
\begin{IEEEdescription}[\IEEEsetlabelwidth{The  Label}] 

  \item [{$G$}] Set of all generation units.
  \item [{$g$}] Index for generation unit.
  \item [{$G(i)$}] Set of all generation units at bus $i \in N$.
  \item [{$i$}] Index for bus.
  \item [{$K$}] Set of all transmission branches.
  \item [{$k$}] Index for transmission branch.
  \item [{$M$}] Set of all measurements.
  \item [{$m$}] Index for measurement.
  \item [{$N$}] Set of all buses.
  
\end{IEEEdescription}

\noindent
	\addcontentsline{toc}{section}{Nomenclature}
	\emph{\textbf{Parameters,  Vectors  and  Matrices}}
\begin{IEEEdescription}[\IEEEsetlabelwidth{The   Label}] 

  \item [{$\alpha$}] Load shift factor.
  \item [{$\alpha^{min}_k$}] The minimum load shift factor that is the start point for transmission asset $k \in K$ to have overflow.
  \item [{$c_g$}] Production cost of unit $g \in G$.
  \item [{$\mathbf{e}$}] $ n_m \times 1 $ vector of measurement noise errors.
  \item [{$\mathbf{H}$}] $ n_m \times n_b $  Jacobian matrix of the system.
  \item [{$\mathbf{H^\prime}$}] $n_b \times n_b$ dependency matrix between power injection measurements and state variables.
  \item [{$\mathbf{H_i^\prime}$}] $i^{th}$ row of $H^\prime$ ($i \in N$).
  \item [{$L_i$}] Active load (MW) at bus $i \in N$.
  \item [{$lb_i$}] Lower bound for load deviation at each bus $i \in N$.
  \item [{$N_1$}] Number of states that can be compromised by attacker.
  \item [{$n_b$}] Number of buses.
  \item [{$n_{br}$}] Number of transmission branches.
  \item [{$n_m$}] Number of measurements.
  \item [{$\bar{P}_g$}] Fixed dispatch point of unit $g \in G$.
  \item [{$P_k^{max}$}] Continuous thermal rating of transmission branch $k \in K$.
  \item [{$P_g^{min}$}] Lower limit on generation capacity of unit $g \in G$.
  \item [{$P_g^{max}$}] Upper limit on generation capacity of unit $g \in G$.
  \item [{$PTDF_{k,i}^R$}] Power transfer distribution factor for branch $k \in K$ and bus $i \in N$ (injection) with regard to reference bus R (withdrawal).
  \item [{$\tau$}] Residual-based bad data detector threshold.
  \item [{$ub_i$}] Upper bound for load deviation at each bus $i \in N$.
  \item [{$\mathbf{Z}$}] $ n_m \times 1 $ vector of measurements.
  
\end{IEEEdescription}

\noindent
	\addcontentsline{toc}{section}{Nomenclature}
	\emph{\textbf{Variables}}
\begin{IEEEdescription}[\IEEEsetlabelwidth{The   Label}] 

  \item [{$\mathbf{c}$}] $ n_b \times 1 $ vector of false data introduced to bus angles by attacker.
  \item [{$\mathbf{H_{i}^{'} c} (\Delta L_i)$}] Active load deviation at bus $i \in N$.
  \item [{$P_g$}] Dispatch point of unit $g \in G$.
  \item [{$P_l$}] Active power flow on target line $l \in K$.
  \item [{$\mathbf{x}$}] $n_b \times 1$ vector of actual state variables.
  \item [{$\mathbf{\hat{x}}$}] $n_b \times 1$ vector of estimated state variables.
  
\end{IEEEdescription}

\section{Introduction}

\setlength{\parindent}{1em}

\IEEEPARstart{I}{n} power systems, state estimation (SE) is one of the key functions of energy management systems (EMSs) since many real-time operational and market decisions are driven by its results. SE is the process of using fields' measurements to estimate systems' state variables with minimum error. 
Due to some limitations, like sensor calibration errors, topology errors, data transfer inaccuracies, and cyber-attacks, received measurements (inputs to SE) are not clean (include noise or false data), which would affect the accuracy of the SE process. To reduce the effect of noisy measurements on the SE process, state estimators are equipped with bad data detectors (BDDs) to flag and remove noisy data. 

False data injection attacks are a class of cyber-attacks that attempt to maliciously change the measurements and interfere in the SE process by targeting the vulnerability of BDDs. BDDs are not looking for intelligent attackers; rather, they are looking for physical limitation driven events\textemdash measurement errors, faulty equipments, etc. Therefore, it would be an easy task for intelligent attackers to bypass BDDs and remain undetectable. The researchers in \cite{liu2011false,kim2013topology,hug2012vulnerability} showed the incapability of BDDs to detect generated FDIAs against both direct current state estimation (DCSE) and alternating current state estimation (ACSE). 
Likewise, they addressed the conditions under which an attacker with complete information about a system could bypass the BDD and remain undetectable. The authors in \cite{ yang2014false} demonstrated that without the assumption of having access to all measurements, launching an FDIA with the least number of measurements to be compromised is an NP-hard problem. To tackle this issue, the authors in 
\cite{hao2015sparse,kosut2011malicious,ozay2013sparse,rahman2012false,anwar2016data,kim2015subspace,tajer2017false} attempted to generate FDIAs with incomplete information about the systems' topology by applying heuristic methods, greedy algorithms, graph-theoretic approaches, and sparse optimization methods. The research study in \cite{yu2015blind} illustrated that even without any information about systems' topology,  attackers could construct undetectable FDIAs.

The focus of this study is on the load-redistribution (LR) attack, which is a way to implement a FDIA against power systems. In LR attacks, the attackers attempt to falsify bus injection measurements to either physically or economically damage the power systems. Various researchers proposed bi-level or attacker-defender optimization problems to model LR attacks with different objectives, like maximizing operation cost or maximizing power flow on a target line \cite{yuan2011modeling,liang2016vulnerability,chu2016evaluating,salmeron2004analysis,liu2016cyber,yuan2012quantitative, kaviani2019identifying}, where the latter is the focus of this study.  
For instance, the attack model in \cite{yuan2011modeling} was designed in a bi-level format, in which the upper level models the attacker’s objective, maximizing the operation cost (generation cost + load shedding cost), and the lower level models the system’s response to the attack based on a base-case security-constrained economic dispatch (SCED). Likewise, the attack models in \cite{liang2016vulnerability} and \cite{chu2016evaluating} were developed in bi-level formats. Their upper level objectives maximized the physical damage of a target line, and their lower levels modeled the systems' response using a nonlinear alternating current optimal power flow (ACOPF) and direct current optimal power flow (DCOPF), respectively. Moreover, the study in \cite{yuan2012quantitative} investigated the physical and economic effects of LR attacks considering both immediate and delayed fashions. For the immediate attacking goal, they proposed a bi-level problem to identify the worst-case attack scenario with an economic goal. For the delayed attacking purpose, they introduced a tri-level problem to maximize the operation cost as a delayed effect of tripping an overloaded line. 

The authors in \cite{tan2017cyber} proposed a bi-level mixed integer linear programming to design an LR attack against multiple transmission assets. LR attacks with incomplete systems' information were designed in \cite{liu2014local} and \cite{liu2015modeling} by finding the best local attacking region. 

Such prior studies have done a great job demonstrating the vulnerability of traditional BDDs, which were previously designed to detect anomalies caused by some physical limitations. It is easy not to be detected when nobody is looking for you or, in other words, ``The greatest trick the devil ever pulled was convincing the world he didn't exist" \cite{movie}.

Now, researchers have acknowledged the existence of attackers and their ability to remain undetectable, which has pushed them to seek a solution. In the first place, standing against intelligent attackers starts by protecting power systems from FDIAs. 
Protection-based actions refer to some preventive actions, which are done at the pre-attack stage, to make it hard for attackers to launch FDIAs against power systems. In this regard, the authors in \cite{yang2017pmu} proposed to place secure phasor measurement units (PMUs) at key buses in the system to defend against FDIAs. 

In \cite{bobba2010detecting}, the authors addressed a way to find the most efficient sets of measurement sensors that need to be protected from the operators' point of view and identified the optimal sets using the brute-force approach. In \cite{deng2017defending}, the authors modeled the problem of finding the least-budget defense strategy as a mixed integer nonlinear programming and applied Benders' decomposition to solve the proposed model. In \cite{esmalifalak2013bad}, the interaction between an attacker and defender is modeled by a two-person zero-sum strategic game where the players attempt to find the Nash equilibrium and maximize their profits, considering the fact that attackers and system operators are not able to attack and defend all measurements. In \cite{kim2011strategic}, the authors proposed a method to find the smallest set of measurements, which provides a protection scheme against the worst-case scenario in which the attack affects the values of the most vulnerable state variables. In \cite{deka2014data}, the authors investigated the graph theory to find the minimum set of measurements that need to be protected. The authors in \cite{kim2011strategic} and \cite{deka2014data} developed their methods based on greedy algorithms for solving NP-hard protection-based problems. The authors in \cite{dan2010stealth} determined the smallest set of protected measurements based on an iterative path augmentation algorithm for both perfect protection and non-perfect protection cases, which refer to protection schemes with zero possibility of hidden attacks and possibility of hidden attacks, respectively. In \cite{hu2017secure}, the authors proposed an algorithm to secure the SE process, as well as a method to reconstruct the attacked signals. However, they focused on a noiseless framework, which is not the case in reality. 

Referring to 
\cite{deka2016jamming}, attackers still could launch an attack even when all measurements have been protected from FDIAs except one of them, which implies the necessity of a detection scheme. Therefore, designing intelligent false data detectors is the next step to stand against intelligent attackers. Various FDIA detection methods were proposed and developed in \cite{gao2016identification,manandhar2014detection,chaojun2015detecting,liu2014detecting,li2015quickest, esmalifalak2017detecting,ozay2016machine,foroutan2017detection,andrea2018machine,he2017real,li2018false} based on various techniques like the Kalman filter, adaptive cumulative sum, low-rank decomposition (LD), Kullback-Leibler distance (KLD), sparse optimization, machine learning, and deep learning.

In this study, we develop a non-probabilistic detection mechanism based on the fundamental knowledge of the  laws  of  physics in power systems to detect LR attacks, which attempt to cause overflows on transmission assets. This is an online monitoring mechanism that allows operators to track load deviations (given a target asset) at each time interval and flag malicious movements.

Our approach differs from other existing methods in different aspects, like, 
\begin{itemize}
    \item Our method successfully detects LR attacks (even the weakened ones) assuming that attackers have no limitation for altering state variables. For instance, in \cite{gao2016identification} and \cite{liu2014detecting}, the authors developed their proposed detection methods based on the assumption that attackers are limited to alter some of the state variables. Moreover, compared to these studies, our proposed method is modeled based on a linear and convex problem. At the same time, the authors in \cite{gao2016identification} and \cite{liu2014detecting} used the matrix low rank decomposition technique to detect false data, which introduces non-linear convex optimization problems with more computational complexity.
    
    \item Our method successfully distinguishes random attacks from normal noise errors. In this regard, the study in \cite{manandhar2014detection} proposed to use the Kalman filter and a Euclidean distance metric to overcome the disability of existing Chi-Square statistic-based detectors to detect FDIAs, which are wisely designed to fit the distribution of historical data or normal noise errors. The authors reported $99.73 \%$ accuracy to filter false positives due to the noise errors, but they set the threshold three times bigger than the standard deviation of the generated random noise errors. Whereas, our method perfectly distinguishes LR attacks scenarios (even the weakened ones) from different samples of Gaussian and non-Gaussian noise errors with realistic assumptions. In other words, our method is successful in detecting unobservable LR attacks, which are wisely created in such a way that the deviations fall into the potential spectrum of generally accepted noise errors but have the preferred values be at preferred buses.
    
    \item Compared to \cite{chaojun2015detecting} and \cite{li2015quickest}, our method has not been developed based on the historical or statistical data; instead, it uses the current data of the system (the SE output) to detect any malicious movement.
    
    \item In \cite{esmalifalak2017detecting,ozay2016machine,foroutan2017detection,andrea2018machine,he2017real}, the authors developed detection mechanisms against cyber-attacks based on machine learning. Machine learning can make cyber security simpler, less expensive, and far more effective. However, it can only do these things when there is a large amount of underlying historical data that provides a complete picture of the environment. However, the proposed method in this study can detect LR attacks regardless of the quantity and quality of the available historical data.
\end{itemize}

The main contributions of our study are summarized as follows:
\begin{enumerate}
 \item Leveraging power systems domain insights to identify an underlying exploitable structure in LR attack problems, which helps operators to predict the attackers' behavior.

\item Mathematically proving the ability of a greedy algorithm to solve the exploitable structure of LR attack problems to optimality, which leads system operators to find the most sensitive buses very fast even for large interconnections. 

\item Proposing the number of proper deviations at sensitive buses (NPDSB) as an index that can detect LR attacks and determine a perfect boundary between LR attacks and normal noise errors.

\item Developing a real-time approach to detect LR attacks, with the goal of causing an overflow on a transmission asset, without significant changes and disruptions in existing EMSs.
\end{enumerate}
This paper is organized as follows. Sec. II presents a short background on DCSE, the condition to launch an undetectable FDIA against DCSE, and LR attacks. Sec. III is divided into two subsections; the first one identifies the exploitable structure of the core problem of more sophisticated LR attack problems and the second one provides a mathematical proof to demonstrate the ability of a greedy algorithm to obtain a global optimum for the identified structure of the core problem. Sec. IV and V present simulation results and concluding remarks, respectively.

\section{Background}

\subsection{DCSE and Undetectable FDIA}\label{A}
In the DCSE process, measurements are linked to state variables (voltage angles) via linear equations. Eq. \ref{DC_SE_formula_1} represents these linear equations in a matrix form:

\begin{equation}\label{DC_SE_formula_1}
\small
\mathbf{Z = Hx + e} ~,
\end{equation}
where $\mathbf{Z}$ is the $n_m \times 1$ vector of measurements, $\mathbf{x}$ is the $n_b \times 1$ vector of actual state variables of the system that needs to be estimated, $\mathbf{H}$ is the $n_m \times n_b$ Jacobian matrix of the system, and $\mathbf{e}$ represents the $n_m \times 1$ vector of measurement noise errors.

A common approach to measure the accuracy of the SE process is to compare the $2$-norm of the measurements residual
with a certain threshold ($\tau$).
Then, if the $2$-norm of the residual for a set of measurements ($\mathbf{Z}$) is greater than $\tau$, it means that $\mathbf{Z}$ contains unacceptable bad data. The $2$-norm of the residual is determined as shown in equation \ref{residual}, where $\mathbf{\Hat{x}}$ is the $n_b \times 1$ vector of estimated states and the $||.||_2$ denotes the $2$-norm of a vector, also known as the Euclidean norm, which calculates the distance of a vector coordinate from the origin of the vector space.
\begin{equation}\label{residual}
\small
\mathbf{||R||_2} = \mathbf{||Z-H\mathbf{\Hat{x}}||_2}.
\end{equation} 

A key theorem in \cite{liu2011false} states that the vector of contaminated measurements $\mathbf{Z_a = Z +a}$, in which vector $\mathbf{a}$ represents the malicious data added to
actual measurements, is able to bypass residual-based BDDs if it is a linear
combination of the column vectors of the Jacobin matrix $\mathbf{H}$. Therefore, the authors in \cite{liu2011false} defined $\mathbf{a = Hc}$, in which $\mathbf{c}$ is the state variable errors' vector, and proved the residual-based BDD deficiency to detect the attack vector $\mathbf{a}$. \\
\textbf{Proof}:
Assume that the vector of estimated state variables after adding vector $\mathbf{a}$ to the actual vector of measurements  $\mathbf{Z}$ is $\mathbf{\Hat{x}_a}=\mathbf{\Hat{x}} + \mathbf{c}$, then the $2$-norm of the residual after the attack is $ \mathbf{||Z_a-H\mathbf{\Hat{x}_a}||_2}$. After substituting $\mathbf{Z_a}$ with $\mathbf{Z+a}$ and $\mathbf{\Hat{x}_a}$ with $\mathbf{\Hat{x}} + \mathbf{c}$, the $2$-norm of the residual is converted to $\mathbf{||Z -H\Hat{x}} + (\mathbf{a -Hc})||_2$. Then, considering the first and main assumption in the theorem ($\mathbf{a = Hc}$),  equation \ref{attackedresidual} is true.

\begin{equation}\label{attackedresidual}
\small
\mathbf{||R_a||} = \mathbf{||R||} = \mathbf{||Z-H\Hat{x}||_2}<\tau.
\end{equation}

\subsection{LR Attacks}\label{B}
Every LR attack starts by falsifying bus injection measurements. In this paper, it is assumed that the attackers in LR attacks avoid changing the measurements related to the generation part since the control center directly communicates with the power plants' control rooms. Moreover, there should not be deviations at zero injection buses. 

In this paper, the only way to damage power systems through an LR attack is to increase the loads at some buses and decrease the loads at other buses. The net load should remain unchanged to avoid frequency issues. Likewise, attackers should modify power flow measurements to follow load deviations. In addition, the load deviation at each bus should be neither
more nor less than pre-determined constant values. If so, the operator would flag that set of load measurements since it has load deviations far from the short-term load forecasting. These constant values are usually determined by a percentage of the forecasted load value at each bus in different directions.

At the end, after generating an undetectable LR attack, the SCED is fed with a contaminated set of loads and provides a set of fake dispatch points that leads the system to an insecure or inefficient operating state.

For instance, Fig. \ref{bi-level-opt} illustrates an example of a bi-level LR attack problem to maximize the flow of a target transmission branch (line $l$) with limited access to specific meters  \cite{chu2016evaluating}. In the upper-level, the attacker attempts to maximize the power flow on the target line subject to the number of available resources ($N_1$) and the limitations on load deviations. The lower-level is a DCOPF that models the system's response to the attack vector generated in the upper-level. 
    
\begin{figure}[!h]
\centering
\includegraphics[trim =37mm 111.5mm 10mm 61mm, clip, width=11.0cm]{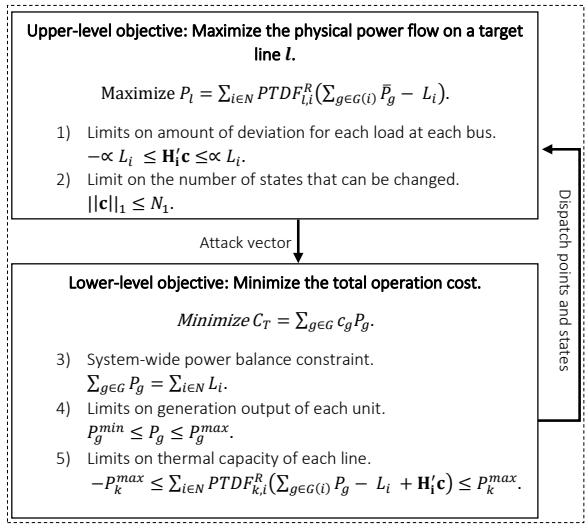}
\caption{A bi-level model for generating an LR attack to physically damage a particular transmission asset.}
\label{bi-level-opt}
\end{figure}    
\vspace{-1.2em}    
\section{Modeling and Methodology}\label{sectionIII}

There are some drawbacks associated with protection-based schemes, such as they reduce measurement redundancy \cite{chaojun2015detecting} and they could not guarantee a perfect protection against FDIAs \cite{deka2016jamming}. Therefore, we proposed a detection mechanism against LR attacks, which is developed based on the deep knowledge of power systems.
To do so, firstly, an exploitable structure for the core problem of bi-level LR attack problems is identified. Then, based on a theorem for the optimality conditions of that exploitable structure, the proposed approach is developed and described.

\subsection{The Identified Exploitable Structure of the Core Problem}\label{3.A}
LR attacks are designed in such a way that they move load measurements up and down so that attackers achieve the maximum physical damage on a target transmission asset. Changing the load pattern will affect power flows. Power transfer distribution factors (PTDFs), or shift factors (SFs), determine the impact of the load change at each bus on a particular transmission asset. For instance, assume an operator wants to remove $20$ MW overflow on a particular line. To do so, the operator will take a generator at a bus with PTDF equals to $0.5$ for that line and move it by $-20$ MW. Then he/she will take a different generator at another bus that has a PTDF for that line at $-0.5$ and move it by $20$ MW. This procedure would result in $-20$ MW and $0$ MW change in the line's flow and total supply, respectively. 

Hence, the trivial approach for an operator who wants to reduce the flow on a particular line, with minimum changes, is to rank all PTDFs with flexible resources from largest to smallest (the most positive to the most negative). 
Then, he/she simply starts reducing the net injection of the resource at the top and simultaneously (MW for MW), increasing the resource at the bottom. If either resource runs out of capacity he/she moves to the next resource on that end and continues until the overflow disappears. 

The essence of the attackers' approach is also the same as the operators' approach. Still, attackers are limited by the number of changes they can apply to the original
resources to avoid being detected. Consequently, problem \ref{Algorithm_3's_Alternative_1}-\ref{Algorithm_3's_Alternative_3} is defined as the core problem of LR attack problems, which attempts to maximize a branch overflow (in a proper direction) relative to the flexibility of resources throughout the system.
\begin{equation}\label{Algorithm_3's_Alternative_1}
\small
\underset{\mathbf{H_{i}^{'} c}}{\text {Maximize}} ~ \pm \sum_{i \in N} (\mathbf{H_{i}^{\prime} c})PTDF_{l,i}^{R} ~,
\end{equation}
\begin{equation}\label{Algorithm_3's_Alternative_2}
\small
s.t. \ \ -\alpha L_i \leq \mathbf{H_{i}^{\prime} c} \leq \alpha L_i ~ i \in N ~, 
\end{equation}
\begin{equation}\label{Algorithm_3's_Alternative_3}
\small
\ \  \ \ \sum_{i \in N} \mathbf{H_{i}^{\prime} c} = 0 ~,
\end{equation}
where $PTDF_{l,i}^{R}$ is the power transfer distribution factor for branch $l \in K$ with respect to the injection at bus $i \in N$ and withdrawal from the reference bus $R$. $\mathbf{H_{i}^{\prime} c}$ is equal to $\Delta L_i$, which is the malicious load deviation at bus $i\in N$. We used $\mathbf{H_{i}^{\prime} c}$ to emphasize the fact that attackers need to change bus angles in order to get appropriate deviations in loads. For instance, if the above problem results in $5$ MW load deviation at bus $2$ ($\mathbf{H_{2}^{\prime} c} = \Delta L_2 = 5$ MW), the attacker needs to design vector $\mathbf{c}$ in such a way that the value of $\mathbf{H_{2}^{\prime} c}$ is equal to $5$ MW. The load shift factor and forecasted load at each bus $i \in N$ are presented by $\alpha$ and $L_i$, respectively.

In problem \ref{Algorithm_3's_Alternative_1}-\ref{Algorithm_3's_Alternative_3}, the main decision variables are the load deviations and the objective is to maximize the overflow on a target transmission asset. 
The $\pm$ notation implies that the overflow direction on a target transmission asset could be both positive and negative, and depends on the target asset's pre-attack power flow direction. For instance, if an attacker wants to cause an overflow on a target asset with the pre-attack power flow equals to $+100$ MW, he/she should use $(+)$ in equation (\ref{Algorithm_3's_Alternative_1}), and if the target asset's pre-attack power flow is $-100$ MW, the attacker should use $(-)$ in equation (\ref{Algorithm_3's_Alternative_1}). 

Constraints in \ref{Algorithm_3's_Alternative_2} impose the deviation at each bus to be neither more nor less than $+ \alpha$ or $- \alpha$ percent of the forecasted load at that bus (they also impose no change at zero injection buses), respectively. Constraint \ref{Algorithm_3's_Alternative_3} ensures that the net load after the LR attack remains unchanged.

In the following, the impacts of two different LR attacks, created by solving problem \ref{Algorithm_3's_Alternative_1}-\ref{Algorithm_3's_Alternative_3}, on a particular line in a 3-bus system (Fig. \ref{3-bus}) is illustrated and the results are presented in Table \ref{3bus_table1}. 

\begin{figure}[!ht]
\centering
\includegraphics[trim =58mm 153mm 11mm 26mm, clip, width=12.0cm]{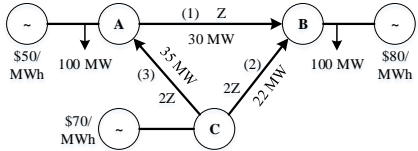}
\caption{3-bus test case diagram.}
\label{3-bus}
\end{figure}

\begin{table}[!ht]
\scriptsize
\caption{The cyber and physical flows associated with two different LR attacks against target line $1$ in the 3-bus test case.}
\label{3bus_table1}
\centering
\begin{tabular}{|c|c|c|c|c|}
\hline
 & \multicolumn{4}{c|}{\textbf{Attack Scenarios}} \\ 
\hline
    & \multicolumn{2}{c|}{$\alpha = \%5$} & \multicolumn{2}{c|}{$\alpha = \%10$} \\
\hline
\multirow{8}{*}{\makecell{\textbf{Cyber} \\ \textbf{Results}}}    & L$_A$ (MW) & $105$   & L$_A$ (MW)  &  $110$  \\ \cline{2-5}
     & L$_B$ (MW)   & $95$  & L$_B$ (MW)  &  $90$  \\ \cline{2-5}
     & L$_{Ref}$  (MW) & $0$  &  L$_{Ref}$  (MW) &  $0$  \\ \cline{2-5}
     & P$_A^g$   (MW) &  $142.5$ & P$_1^g$  (MW) &  $147.5$  \\  \cline{2-5}
     & P$_B^g$  (MW) & $57.5$  & P$_2^g$  (MW) & $52.5$  \\  \cline{2-5}
     & P$_C^g$  (MW) & $0$  & P$_3^g$  (MW) &   $0$ \\  \cline{2-5}
     &  P$_1^{target}$  (MW) & $30$  & P$_1^{target}$  (MW) &  $30$  \\  \cline{2-5}
     &  Cost ~(\$) & $11725$  & Cost (\$)  &  $11575$  \\ \hline
      \multirow{3}{*}{\makecell{\textbf{Physical} \\ \textbf{Results}}} &  P$_1^{target}$ (MW) & $34 > 30$ & P$_1^{target}$ (MW) &  $38 > 30$  \\ \cline{2-5}
     & P$_2^k$  (MW) &  $-8.5$ & P$_2^k$  (MW)  & $-9.5$   \\ \cline{2-5}
     & P$_3^k$  (MW)&  $8.5$ & P$_3^k$  (MW) & $9.5$   \\  
\hline
\end{tabular}
\end{table}

There are generation units at all buses and bus C is the reference bus. The minimum and maximum capacities of all three units are $0$ and $150$ MW, respectively. Line $1$ is the target line, and two different attack vectors were created based on two different load shift factors ($5 \%$ and $10 \%$). To get the results in Table \ref{3bus_table1}, we: 1) solved problem \ref{Algorithm_3's_Alternative_1}-\ref{Algorithm_3's_Alternative_3} to find the most damaging attack vectors, 2) used the generated attack vectors to falsify the original loads, 3) ran the DCOPF problem to achieve the fake dispatch points, and 4) ran the DC power flow (DCPF), considering the actual loads and fake dispatch points, to find the actual physical power flows on all transmission lines. 

In case $1$, when $\alpha$ was $ 5 \%$, malicious deviations ($\Delta L_A$ = $+5$ MW, $\Delta L_B$ = $-5$ MW) led the DCOPF to provide a set of fake dispatch points (P$_A^G$ = $142.5$ MW, P$_B^G$ = $57.5$ MW, and P$_C^G$ = $0$). Considering the actual loads (L$_A$ = $100$ MW, L$_{B}$ = $100$ MW), this set of fake dispatch points caused P$_1^{target}$ = $34$ MW, P$_2^k$ = $-8.5$ MW, and P$_3^k$ = $8.5$ MW, which showed $13.3 \%$ overflow on line $1$. In case $2$, when $\alpha$ was $ 10 \%$, all simulations were repeated. This time the physical line flows were P$_1^{target}$ = $38$ MW, P$_2^k$ = $-9.5$ MW, and P$_3^k$ = $9.5$ MW, which showed $26.6 \%$ overflow on line $1$. 

The results demonstrate that as the attack's energy increases ($\alpha$ increases), the damage could be more significant, which at some point in time could cause the target line trips offline and results in a cascading blackout. However, there should be a trade-off between the attack's energy and the detection probability since as the energy increases, the detection probability increases. 

In this study, linear optimal power flow models have been considered; this work is extendable for non-convex ACOPF formulations since the underlying special structure in the classical DCOPF is caused by Kirchhoff's Voltage Law (KVL) and Kirchhoff's Current Law (KCL), which remain present in all optimal power flows (OPFs).

\subsection {Proving the Application of a Greedy Algorithm to Solve the Core Problem of the LR Attack}
After identifying the special structure of the core problem, the next step is to prove that this problem, which is a variant of the fractional knapsack problem \cite{ferdosian2016greedy} from an operations research perspective, can be solved to optimality with a greedy algorithm. Hence, in this part, the ability of greedy algorithms to optimize problem \ref{Algorithm_3's_Alternative_1}-\ref{Algorithm_3's_Alternative_3} has been proved and presented.

Greedy methods attempt to build up a solution for a mathematical problem by making a sequence of choices. These choices depend on each other, and the previous choices in the solving process affect the other decisions that can be made later in the process. 
Considering the values of possible choices at each step, a greedy algorithm selects the best local choice. This choice is called a greedy choice, and the resulting algorithm is called a greedy algorithm. Greedy algorithms produce good solutions for some mathematical problems. For example, it provides the global optimum for the fractional knapsack problem \cite{ferdosian2016greedy}.

In the following, a mathematical proof is presented to demonstrate that a greedy algorithm can solve problem \ref{Algorithm_3's_Alternative_1}-\ref{Algorithm_3's_Alternative_3} to optimality. 

After applying a greedy algorithm to solve this problem, at least one of the decision variables ($\Delta L_i$) is either at its lower bound $(l_i)$ or upper bound $(u_i)$, so optimality follows from the theorem below.

\begin{thm}
Feasible solution $(\Delta L_1 ,  ... , \Delta L_{n_b})$ is optimal if and only if, whenever $PTDF_{l,i}^R > PTDF_{l,j}^R$, we find that $\Delta L_i=u_i$ or $\Delta L_j=l_j$ (or both).
\end{thm}
\begin{proof}
$\rightarrow$ Suppose by contradiction that there is an optimal solution for which $PTDF_{l,i}^R > PTDF_{l,j}^R$, $\Delta L_i<u_i$, and $\Delta L_j > l_j$. Compute  $\delta = min(u_i-\Delta L_i,\Delta L_j-l_j)$. Then, add $\delta$ to $\Delta L_i$ and subtract it from $\Delta L_j$, which gives another feasible solution. However, $\sum_{t}^{n_b} \Delta L_tPTDF_{l,t}^R$ increases  by $\delta(PTDF_{l,i}^R-PTDF_{l,j}^R)$, which is positive. Hence the solution cannot be optimal.

\begin{figure*}[!b]
\centering
\includegraphics[trim =20mm 0mm 6mm 21mm, clip, width=16cm, height=8.05cm]{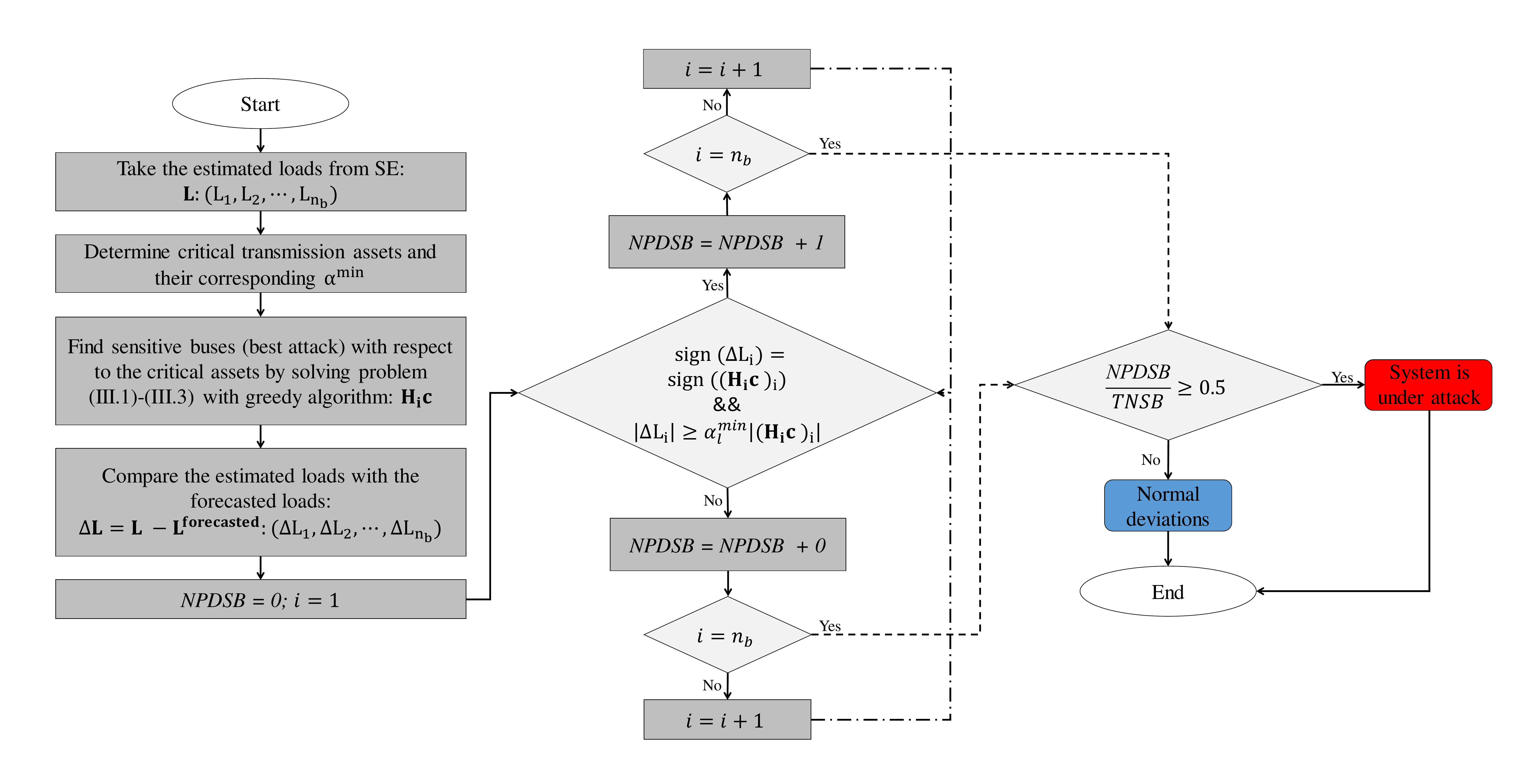}
\caption{The proposed LR attack detection flowchart.}
\label{Flowchart}
\end{figure*}

$\leftarrow$ Suppose by contradiction $\mathbf{S} = (\Delta L_1, ... , \Delta L_{n_b})$ is a feasible solution for which whenever $PTDF_{l,i}^R > PTDF_{l,j}^R$, $\Delta L_i = u_i$ or $\Delta L_j = l_j$ but is not an optimal solution. Choose an optimal solution $\mathbf{O} = (y_1 , ... , y_{n_b})$ in which the number of times that $\Delta L_t \neq y_t$, ($t \in N$) is as small as possible. Note that $\sum_{t}^{n_b} y_tPTDF_{l,t}^R > \sum_{t}^{n_b} \Delta L_tPTDF_{l,t}^R$. Because $\sum_{t}^{n_b} y_t = \sum_{t}^{n_b} \Delta L_t = 0$ there is an item $a$ for which $y_a > \Delta L_a$ and another item $b$ for which $y_b < \Delta L_b$. It follows that $\Delta L_a < u_a$ and $\Delta L_b > l_b$ (by the conditions that S satisfies), and hence that $PTDF_{l,a}^R \leq PTDF_{l,b}^R$. Let $\delta = min(y_a-\Delta L_a,\Delta L_b-y_b)$. In $\mathbf{O}$, subtract $\delta$ from $y_a$ and add $\delta$ to $y_b$, to get a feasible solution $\mathbf{O'}$ that changes $\sum_{t}^{n_b} y_tPTDF_{l,t}^R$ by $\delta(PTDF_{l,b}^R-PTDF_{l,a}^R)$. Now if $PTDF_{l,a}^R<PTDF_{l,b}^R$, $\mathbf{O'}$ yields a larger sum that does $\mathbf{O}$; this contradicts the optimality of $\mathbf{O}$. So, this must mean that $PTDF_{l,a}^R = PTDF_{l,b}^R $. Then $\mathbf{O'}$ is also optimal but, by construction, has fewer items than $\mathbf{O}$, in which it disagrees with $\mathbf{S}$; this contradicts the requirement that $\mathbf{O}$ is an optimal solution with fewest such differences. Therefore, it is concluded that no such $\mathbf{O}$ can exist, and hence that $\mathbf{S}$ is optimal.
\end{proof}

This proof for global optimality results in developing a mechanism to predict the attackers' moves and find the sensitive buses, given a target transmission asset. In fact, this proof shows that the identified structure of the core problem is solvable by a trivial sorting approach, and there is no reason to solve any complicated problem for operators to detect LR attacks, since the attackers' strategies are strikingly simple and trivial for this type of attacks. By using this mechanism, operators can swiftly determine the sensitive buses for any critical transmission asset and track the deviations at those buses to flag any set of changes that contributes to overload that asset. Consequently, achieving a global solution (also near optimality) becomes impossible for attackers. 
Therefore, attackers have to introduce some form of randomness to avoid being detected. It is predictable that even though attackers can create a situation where randomness is applied to their strategy, so much of the feasible space is cut off by the proposed attack detection mechanism, and the impacts of this class of attacks is rendered to be very low.

After finding the sensitive buses associated with all critical transmission assets, the NPDSB determines whether the current set of estimated loads is malicious or normal. To find the NPDSB related to each set of deviations, operators need to check both direction and magnitude of deviations at sensitive buses. Checking the direction of each deviation is simple and straightforward. For instance, if the greedy algorithm results in the load at sensitive bus $i$ to increase, but the deviation related to the current set of estimated loads at sensitive bus $i$ is negative, regardless of the deviation magnitude, it does not count as a proper deviation at sensitive bus $i$. In the next step, the deviation magnitude should be checked; a deviation with a small magnitude could not contribute significantly and it might be a normal noisy deviation. To do so, operators need to have an appropriate threshold value to differentiate malicious and significant deviations from normal noise errors.
This threshold value is different for each transmission asset in the system. However, considering a system with five critical assets that can be overloaded, as $\alpha$ can be a maximum of $10 \%$, calculating the threshold values for these five assets is sufficient. In this regard, we defined a factor of the forecasted load at each bus as the threshold value for deviation at that bus; if the deviation at a bus is more than its threshold value, it counts as a proper deviation and vice versa. In this study, we used $\alpha_{l}^{min}$ as the factor of the forecasted load to find the magnitude threshold for each critical transmission asset. This is the minimum value of the load shift factor that causes an overflow on asset $l$. This policy was made based on the trivial fact that every attack, which is created with $\alpha$ less than $\alpha_{l}^{min}$, is not an effective attack to cause an overflow on the target asset, and no attacker attempts to attack a system without any damage.

At this point, after finding the NPDSB associated with the current set of loads, operators should make a decision based on the value of the NPDSB and categorize the current deviations either as malicious or normal. Therefore, another threshold for the value of the NPDSB should be defined, which enables operators to determine whether the system is under an LR attack from the viewpoint of flow violations. In this study, if the NPDSB associated with a set of deviations is more than the half of the total number of sensitive buses (TNSBs), which is the number of buses with the PTDF values more than $0.01$ (the cut-off value for PTDFs), then that set of deviations is flagged as a malicious set. The flowchart in Fig. \ref{Flowchart} gives a better view of how the proposed detection mechanism works.

\section{Simulation and Results}
\vspace{-0.5em}
\subsection{Illustrative Test Case}\label{resultsA}
Here, for more clarification, the small IEEE 6-bus test case, shown in Fig. \ref{6Bus_Bestattack&A1}, is used to illustrate the gains from our proposed approach. In this experiment, we generated 
two random vectors ($\mathbf{a_1}$ and $\mathbf{a_2}$) in such a way that one of them is an attack vector and the other is not; then, we used our proposed mechanism to find the attack vector.

Both vectors are samples from a normal distribution, but the one that is the attack is simply arranged in such a way that causes an overload on the vulnerable line from bus $3$ to bus $5$ (line 3-5); this is the basic technique of an unobservable attack: have the deviations fall into the potential spectrum of generally accepted noise error but have the preferred values be at preferred buses. All required information including the load at each bus, PTDFs with respect to line 3-5, and vector $\mathbf{a_1}$ and $\mathbf{a_2}$ are shown in Table \ref{6bus_table1}.

The detection process starts by solving problem \ref{Algorithm_3's_Alternative_1}-\ref{Algorithm_3's_Alternative_3} to optimality using the greedy algorithm (Algorithm \ref{greedy_algorithm}) to find the best attack vector against line 3-5, and also determines the most sensitive buses associated with this line. 
In Algorithm 1, all buses are sorted based on their PTDF absolute values in descending order. Then, considering alpha $\alpha = 10 \%$, the maximum possible deviation is assigned to each bus from the top to the bottom, and constraints \ref{Algorithm_3's_Alternative_2} and \ref{Algorithm_3's_Alternative_3} are imposed on each step.  The proposed algorithm sorts the buses based on their values for the target asset (PTDFs), and then it uses a for-loop to find the optimum point. The main time-consuming step is to sort all of the items in decreasing order of their values. If the buses are already arranged in the required order, then for-loop takes O(N) time (N is the number of variables). Otherwise, since the average time complexity of the sorting step is O(NlogN), the total time is O(NlogN).

The best attack vector, PTDFs, vector $\mathbf{a_1}$, and vector $\mathbf{a_2}$ are reordered and shown in Table \ref{6bus_table2}.

As shown in Table \ref{6bus_table2}, considering $0.05$ as the cut-off value for PTDFs, the TNSB for line 3-5 is $4$ (buses 3, 6, 5, and 2). Likewise, $\alpha_{3-5}^{min}$ is $0.035$ (assuming that flow deviation more than $0.3$ MW cause an overflow), which means that the magnitude of each bus's deviation should be compared with $3.5$ percent of the forecasted load at that bus. 

Vector $\mathbf{a_1}$ has deviations at three buses with proper directions and magnitudes to cause an overload on line 3-5, which means that the NPDSB of this vector is $3$ (buses $3$, $6$, and $5$). On the other hand, vector $\mathbf{a_2}$ has one deviation with proper direction and magnitude (bus $3$), which implies that the NPDSB of this vector is $1$. Consequently, the ratio of the NPDSB of vector $\mathbf{a_1}$ to the TNSB is $0.75$, so it is flagged as an attack (considering $0.5$ as the threshold for NPDSBs). However, this ratio for vector $\mathbf{a_2}$ is $0.25$, which implies that vector $\mathbf{a_2}$ is not a malicious set of deviations.

\renewcommand{\arraystretch}{1}
\begin{table}[H]
\scriptsize
\centering
\caption{Two randomly generated vectors representing net injection deviations, original loads, PTDFs: ordered by bus number.}
\begin{tabular}{|c|c|c|c|c|}
\hline
\textbf{Bus} & \textbf{load (MW)}& \textbf{PTDF} & $\mathbf{a_1}$ \textbf{(MW)}& $\mathbf{a_2}$ \textbf{(MW)}\\ 
\hline
$1$ & $10$ & $0$ & $-0.456$ & $0.976$\\ 
\hline
$2$ & $15$ & $0.062$ & $-0.127$ & $-0.954$ \\
\hline
$3$ & $15$ & $0.289$ & $1.136$ & $1.143$ \\
\hline
$4$ & $30$ & $0.0183$ & $-0.564$ & $-2.051$ \\
\hline
$5$ & $20$ & $-0.1207$ & $-0.751$ & $1.519$ \\
\hline
$6$ & $10$ & $0.152$ & $0.762$ & $-0.633$ \\
\hline
\end{tabular}
\label{6bus_table1}
\end{table}

\captionsetup[algorithm]{font=footnotesize}
\begin{algorithm}[H]
\scriptsize
\caption{The greedy algorithm that is used to optimize problem \ref{Algorithm_3's_Alternative_1}-\ref{Algorithm_3's_Alternative_3}.}
\textbf{Input:} The SE outputs and forecasted loads.  \\
\textbf{Output:} A vector including $n_b$ deviations associated to each bus $i \in N$ $(x[i])$.
\begin{algorithmic}[1]
\State $X \gets 0$;
\State $x[i]\gets u_b[i]$;
\For{\texttt{$i \gets 1 ~ to ~ N$}}
\State $Sorted~PTDF \gets Sort~ PTDFs ~in ~a~ descending~ order$;
\State $X \gets X + u_b[i]$;
\EndFor
\For{\texttt{$i \gets $ Sorted-PTDF indexes}}
\If{$Flow_{[target-line]} \leq 0$}
\If{$(x[i]-l_b[i]) \leq X$} 
\State $X \gets X - (x[i]-l_b[i])$;
\State $x[i] \gets l_b[i]$;
\Else 
\State $x[i] \gets (x[i] - X)$;
\State $X \gets 0$;
\EndIf
\Else 
\If{$(x[i]-l_b[i]) \leq X$} 
\State $X \gets X - (x[i]-l_b[i])$;
\State $x[i] \gets u_b[i]$;
\Else
\State $x[i] \gets (X - x[i])$;
\State $X \gets 0$;
\EndIf

\EndIf
\State \textbf{return} $x[i]$
\EndFor
\end{algorithmic}
\label{greedy_algorithm}
\end{algorithm}

\begin{table}[!ht]
\scriptsize
\centering
\caption{Two randomly generated vectors representing net injection deviations, original loads, PTDFs, the best attack vector, and the deviation thresholds: ordered based on the PTDF values.}
\begin{tabular}{|c|c|c|c|c|c|c|} \hline
\multirow{3}{*}{\makecell{\textbf{Bus}}} &  
     &  
\multirow{3}{*}{\makecell{\textbf{PTDF}}} & 
\textbf{Best} & 
     & 
 \multirow{3}{*}{\makecell{$\mathbf{a_1}$\textbf{(MW)}}} & \multirow{3}{*}{\makecell{$\mathbf{a_2}$\textbf{(MW)}}}  \\
 & \textbf{load} &   & \textbf{Attack} & \textbf{Magnitude} & &  \\
 &  \textbf{(MW)} &  & \textbf{(MW)} & \textbf{Threshold}  & &  \\ 
\hline
$3$ & $15$ & $0.289$ & $1.5$ & $0.525$ & $1.136$ & $1.143$  \\  \cline{1-7}
$6$ & $10$ & $0.152$ & $1.0$ & $0.350$ & $0.762$ & $-0.633$ \\  \cline{1-7}
$5$ & $20$ & $-0.120$& $-2.0$& $0.700$ & $-0.751$& $1.519$  \\  \cline{1-7}
$2$ & $15$ & $0.062$ & $1.5$ & $0.525$ & $-0.127$& $-0.954$ \\  \cline{1-7}
$4$ & $30$ & $0.018$ & $-1.0$& $1.050$ & $-0.564$& $-2.051$ \\  \cline{1-7}
$1$ & $10$ & $0.000$ & $-1.0$& $0.350$ & $-0.456$& $0.976$  \\ 
\hline
\end{tabular}
\label{6bus_table2}
\end{table}

\begin{figure}[!ht]
\centering
\includegraphics[trim =51mm 220.5mm 9mm 27.5mm, clip, width=11.0cm]{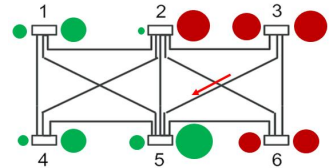}
\caption{6-Bus test case diagram including the attackers' preferred load deviations spectrum and vector $\mathbf{a_1}$'s load deviations spectrum: the left-hand side circles are related to $\mathbf{a_1}$ and the right-hand side circles are related to the best attack vector.}
\label{6Bus_Bestattack&A1}
\end{figure}

Fig. \ref{6Bus_Bestattack&A1} visually displays and compares the best attack vector with random
attack vector $\mathbf{a_1}$. The circle on the left-hand side of each
bus is related to vector $\mathbf{a_1}$, and the circles on the right-hand side are related to the
best attack vector. The size of each circle indicates the load deviation magnitude (a
larger circle implies a more sensitive bus), and the color of each circle indicates the
load deviation direction (red circles mean positive deviations and green circles
mean negative deviations).

\subsection{Case Study on the Modified 2383-Bus Polish Test System} \label{resultsB}
In this case, we evaluated the scalability of our detection
scheme by using a modified version of the 2383-Bus Polish Test System \cite{zimmerman1997matpower2}. The modifications include: decreasing the
line continuous thermal ratings to create base case attacks and setting the negative loads to zero.
In this section, we did multiple evaluations to show the promising features of our proposed detection approach. First, we generated two attack vectors for line $169$ by solving problem  \ref{Algorithm_3's_Alternative_1}-\ref{Algorithm_3's_Alternative_3} two times. The first time we solved this problem with a commercial optimization package and the second time with the proposed greedy algorithm. The goal of this experiment was to numerically demonstrate the ability of the proposed greedy algorithm to find the global solution for problem \ref{Algorithm_3's_Alternative_1}-\ref{Algorithm_3's_Alternative_3} and get the same results as the commercial solver. 
Second, we analyzed the effectiveness and efficiency of the generated attack vector by showing the power flows in the control room that operators see and the actual physical power flows after the attack. 
Third, we demonstrated the ability of the proposed mechanism to detect random LR attacks and distinguish them from both Gaussian and non-Gaussian noise errors.

\subsubsection{Solving Problem \ref{Algorithm_3's_Alternative_1}-\ref{Algorithm_3's_Alternative_3} by Two Methods}\label{B.1}
In this subsection, we demonstrated the ability of the proposed algorithm to solve the special structure of the core problem to optimality by comparing its results with the results of solving problem \ref{Algorithm_3's_Alternative_1}-\ref{Algorithm_3's_Alternative_3} by a commercial optimization package (GUROBI \cite{gurobi}).
To do so, we solved problem \ref{Algorithm_3's_Alternative_1}-\ref{Algorithm_3's_Alternative_3} for line $169$, considering $\alpha$ equals to $10 \%$. Both simulations were run in JAVA on an Intel(R) Xeon(R) CPU with 48 GB of RAM. The attack vectors from both methods perfectly matched each other. Fig. \ref{load_deviations} shows the false deviations associated with some of the most sensitive buses.

\begin{figure}[!ht]
\centering
\includegraphics[trim =18mm 66mm 0.5mm 13mm, clip, width=9.5cm]{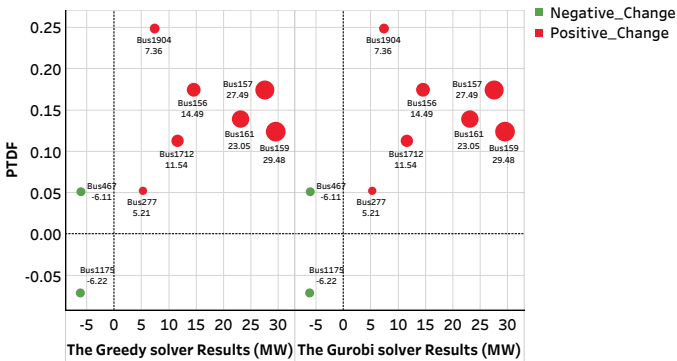}
\caption{Load deviations at some of the most sensitive buses with respect to the target line $169$ in the 2383-Bus Polish system; achieved by solving problem \ref{Algorithm_3's_Alternative_1}-\ref{Algorithm_3's_Alternative_3} using both GUROBI and greedy solvers.}
\label{load_deviations}
\end{figure}

\subsubsection{Attack Efficiency Analysis}
Here, similar to the example in subsection \ref{3.A}, we applied the attack vector to the initial load and ran the DCOPF and DCPF to find both cyber and actual physical power flows on the target line. The results showed that the attack was effective and efficient to cause an overflow.

Here, by effectiveness and efficiency, we meant that the created attack has enough energy to damage the target line. Table \ref {pre-postflows} provides the target line power flow results after the attack, including the cyber flow (the control room flow), physical flow, and the amount of overflow.

\begin{table}[!h]
\scriptsize
\centering
\caption{The control room flow, actual physical flow, and the overflow on line $169$ after the best LR attack against line $169$.}
\begin{tabular}{|c|c|}
\hline
\textbf{Line No.}   & 169   \\
\hline
\textbf{Continuous Thermal Rating (MW)}  & 926.62 \\
\hline
\textbf{Cyber Power Flow (MW)} & -926.62\\
\hline
\textbf{Physical Power Flow (MW)}  & -1178.136\\
\hline
\textbf{Overflow (MW)}  & 251.516\\
\hline
\end{tabular}
\label{pre-postflows}
\end{table}

As shown in Table \ref {pre-postflows}, the generated attack was successful in causing $251.516$ MW overflow on the target line, where it was undetectable for the system operator who saw $926.62$ MW power flow on this line, which is not more than its continuous thermal rating.

\subsubsection{Detection Mechanism Efficiency Analysis}
In this section, we investigated and analyzed the ability and success rate of the proposed method to detect some random weakened LR attacks and distinguish them from noise errors. 
To do so, we broke this subsection down into two parts. First, we did some experiments by targeting line $169$, and second, we took line $251$ as the target line, and repeated all the tests, similar to line $169$. Then, we created different sets of random LR attacks and Gaussian/non-Gaussian noise errors against both target lines while demonstrating the physical effect versus the NPDSB of each set.

\paragraph{Line 169}
This line is categorized as a critical line since there is at least one scenario of LR attack with $\alpha$ at most $10 \%$ that could make this line physically overloaded. Its continuous thermal rating is $926.62$ MW, pre-attack power flow is negative (we had to use $-$ in equation \ref{Algorithm_3's_Alternative_1} to find the best attack), TNSB is $1168$ (cut-off value is $0.05$), and $\alpha^{min}_{169}$ is $0.0425$.

To validate the capability of the proposed method in order to detect and distinguish random LR attacks from normal noise errors, we did two experiments. First, we generated $1000$ random LR attack vectors and compared their physical effects with the physical effects of $1000$ random Gaussian noise errors. Second, we generated $1000$ random LR attacks and $1000$ random Cauchy noise errors (non-Gaussian) and, similar to the first experiment, compared their physical effects.

\begin{table*}[hb]
\centering
\caption{Detailed results associated with two scenarios of random LR attack and two scenarios of normal noise errors for lines $169$ and $251$, including the actual physical flows, control room flows, and ratios of NPDSB to the TNSB.}
\label{my-table}
\begin{tabular}{|c|c|c|c|c|c|c|c|c|c|c|c|}
\hline
\multicolumn{4}{|c|}{\textbf{}} & \multicolumn{4}{c|}{\textbf{Random Attack}} & \multicolumn{4}{c|}{\textbf{Normal Noise}} \\ \hline
\textbf{\begin{tabular}[c]{@{}c@{}}Line \\ No.\end{tabular}} & \textbf{\begin{tabular}[c]{@{}c@{}}$\mathbf{P_k^{max}}$\end{tabular}} & \textbf{TNSB} & \textbf{\begin{tabular}[c]{@{}c@{}}Case \\ No.\end{tabular}} & \textbf{NPDSB} & \textbf{\begin{tabular}[c]{@{}c@{}}Physical \\ Flow\\  (MW)\end{tabular}} & \textbf{\begin{tabular}[c]{@{}c@{}}Cyber \\ Flow (MW)\end{tabular}} & $\mathbf{\frac{NPDSB}{TNSB}}$ & \textbf{NPDSB} & \textbf{\begin{tabular}[c]{@{}c@{}}Physical \\ Flow \\ (MW)\end{tabular}} & \textbf{\begin{tabular}[c]{@{}c@{}}Cyber \\ Flow (MW)\end{tabular}} & $\mathbf{\frac{NPDSB}{TNSB}}$ \\ \hline
\multirow{2}{*}{$169$} & \multirow{2}{*}{$926.62$} & \multirow{2}{*}{$1168$} & $1$ & $762$ & $-1054.2$ & $-901.3$ & $0.65$ & $18$ & $-829.4$ & $-825.6$ & $0.015$ \\ \cline{4-12} 
 &  &  & $2$ & $771$ & $-981.2$ & $-884.1$ & $0.66$ & $29$ & $-846.5$ & $-832.9$ & $0.024$ \\ \hline
\multirow{2}{*}{$251$} & \multirow{2}{*}{$387.34$} & \multirow{2}{*}{$998$} & $1$ & $632$ & $-417.5$ & $-314.9$ & $0.63$ & $16$ & $-287.5$ & $-287.8$ & $0.016$ \\ \cline{4-12} 
 &  &  & $2$ & $623$ & $-395.2$ & $-304.2$ & $0.62$ & $14$ & $-287.3$ & $-286.2$ & $0.014$ \\ \hline
\end{tabular}
\end{table*}

To achieve each random attack vector, we solved problem \ref{Algorithm_3's_Alternative_1}-\ref{Algorithm_3's_Alternative_3}, and each time $\alpha$ ($10 \%$) was multiplied to a random number between $0.425$ and $1$\textemdash based on the fact that for $\alpha$ less than $0.0425$ there is no successful attack with enough energy to damage line $169$. Next, we added a constraint to force the deviations at $150$ randomly selected sensitive buses (for line $169$) to be zero.

We generated Gaussian random noise vectors from a Gaussian distribution with $\mu =0$ and $\delta = \alpha L/3.1$ in such a way that the deviation at each bus was limited to $\alpha$ percent of the forecasted load in either directions. There was no change at zero injection buses, and the net load change in the system was very small. Moreover, we extracted the Cauchy noise vectors from a Cauchy distribution with location $x_0 =0$ and scale $\gamma = 0.01L/3.1$ \cite{cauchy}. All random Cauchy noise errors were created and subjected to the same three constraints, which were applied to the process of creating random Gaussian noise errors.

Fig. \ref{NoiseAttackSeparation169} demonstrates different physical flows on the target line associated with each set versus their respective NPDSBs. This figure includes two sub-figures, where sub-figure (a) shows the comparison between $1000$ sets of random LR attacks and $1000$ sets of Cauchy noise errors and sub-figure (b) shows the same comparison for another $1000$ sets of random LR attacks and $1000$ sets of Gaussian noise errors.
To get the physical power flows in Fig. \ref{NoiseAttackSeparation169}, we followed the same procedure for the 3-bus system example in subsection \ref{3.A}.

As illustrated in Fig. \ref{NoiseAttackSeparation169}, our method flagged $100 \%$ of the $2000$ scenarios of random LR attacks against line $169$ (red points); all points with NPDSBs more than $1168 \times 0.5$ were considered as malicious movements. Likewise, considering both random Gaussian and non-Gaussian noise errors, the results validated the accuracy of the proposed method to differentiate random attacks from noise errors (blue points).

\begin{figure}[!ht]
\centering
\includegraphics[trim =5mm 1mm 0.5mm 5mm, clip, width=8.5cm]{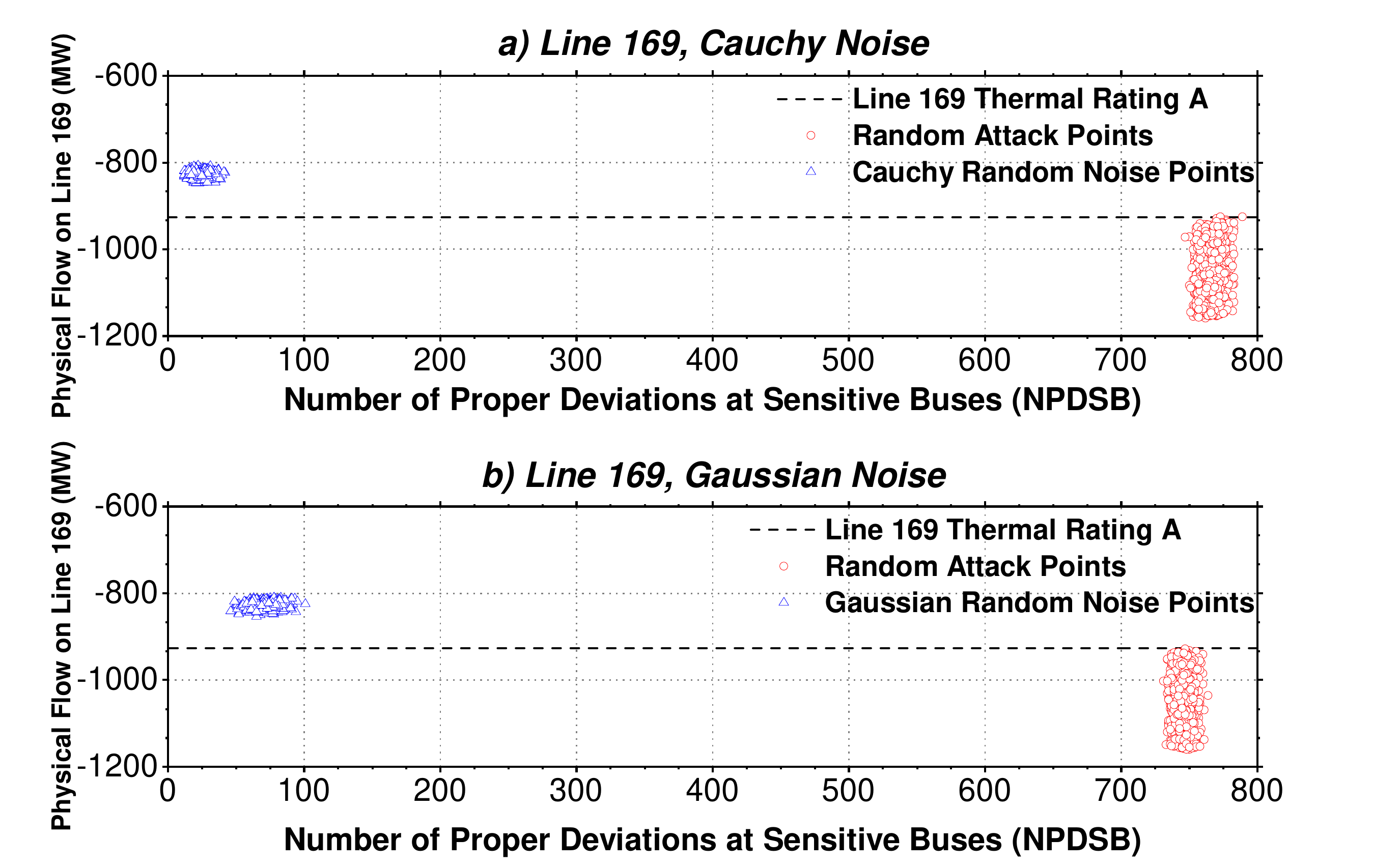}
\caption{Physical effects of different scenarios of load deviations on line $169$ versus the NPDSBs. Sub-figure $a)$ shows the comparison of the physical effects of $1000$ random LR attacks with $1000$ random Cauchy errors ($x_0 =0, ~ \gamma = 0.01L/3.1$) and sub-figure $b)$ shows the comparison of the physical effects of $1000$ random LR attacks with $1000$ random Gaussian errors ($\mu =0, ~ \delta = \alpha L/3.1$).}
\label{NoiseAttackSeparation169}
\end{figure}

\paragraph{Line 251}
Similar to line $169$, line $251$ is a critical line with the continuous thermal rating at $387.37$ MW, TNSB at $998$ (cut-off value is $0.05$), $\alpha^{min}_{251}$ at $0.0686$, and negative pre-attack power flow.

In this part, we did all the simulations that we did for line $169$ to evaluate the functionality of the proposed method to detect LR attacks against another target line. As illustrated in Fig. \ref{NoiseAttackSeparation251}, our scheme successfully detected all $2000$ scenarios of random LR attack (red points); every point with the NPDSB more than $998 \times 0.5$ was flagged as a malicious movement. Furthermore, the proposed method distinguished both types of noise errors from the random LR attacks against this target line, even those, which could not cause an overflow. 

According to the results, there were some random attack scenarios, which had
not enough energy to cause an overflow on the target line (red points above the line
related to continuous thermal rating). It is because some of the randomly selected
buses with zero deviations were among the most sensitive buses. Although these scenarios were not successful in causing an overflow on the target line, the proposed method 
flagged them as a malicious movement since their NPDSBs were more than the determined thresholds. We generated these scenarios to show our method's capability, while this may not be the case in reality.

\begin{figure}[!ht]
\centering
\includegraphics[trim =8mm 2.5mm 0.5mm 5mm, clip, width=8.5cm]{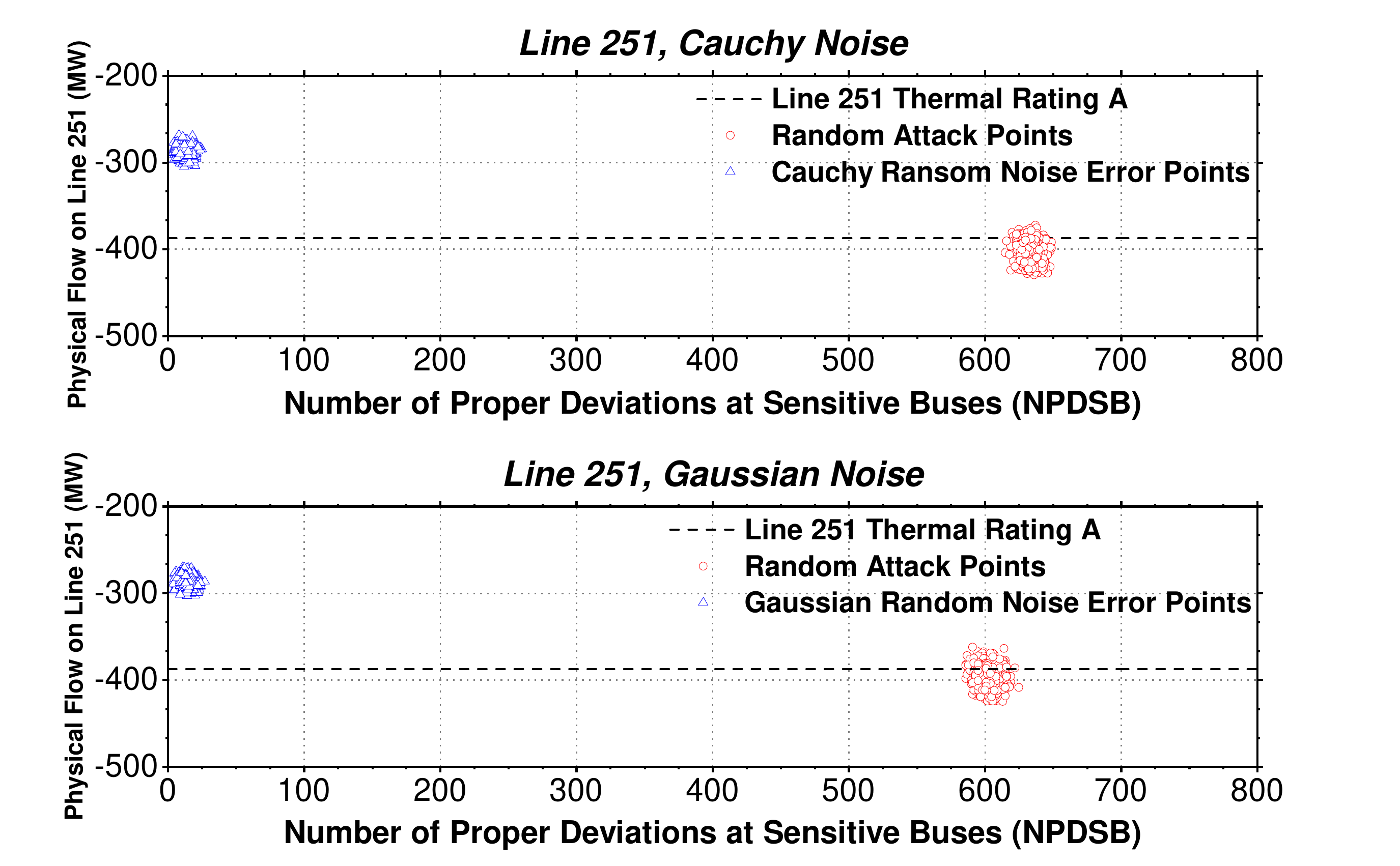}
\caption{Physical effects of different scenarios of load deviations on line $251$ versus the NPDSBs. Sub-figure $a)$ shows the comparison of the physical effects of $1000$ random LR attacks with $1000$ random Cauchy errors ($x_0 =0, ~ \gamma = 0.01L/3.1$) and sub-figure $b)$ shows the comparison of the physical effects of $1000$ random LR attacks with $1000$ random Gaussian errors ($\mu =0, ~ \delta = \alpha L/3.1$).}
\label{NoiseAttackSeparation251}
\end{figure}

Due to the large number of scenarios and limited space, all the results were depicted in Fig. \ref{NoiseAttackSeparation169} and Fig. \ref{NoiseAttackSeparation251}. Additionally, we presented detailed results for eight scenarios in Table \ref{my-table}.

All attack scenarios in Table \ref{my-table} made the target lines physically overloaded. Although 
there was no overflow on either of the target lines based on the cyber power flows, our method successfully flagged all attack cases since the ratios of their NPDSBs to the TNSBs were more than the determined thresholds. For example, in case $1$ for line $169$, the attack could successfully cause a physical overflow on line $169$ ($13.7 \%$ overflow), but the cyber power flow in the control room showed $-901.3$ MW, which is within the thermal limits of the line. This value of cyber power flow prevents the operator from being notified of the existence of an attack in the system. However, for an EMS equipped with our proposed method, even though the control room power flow shows a secure point of operation, the operator can flag the attack since the ratio of the NPDSB to the TNSB is $0.65$. 

On the other hand, the ratios of all the NPDSBs to the TNSBs associated with random noise scenarios were less than the thresholds, based on which our method did not flag these sets of load deviations. Likewise, the physical power flows on all transmission assets were all within their capacity limits, which confirmed the decisions made by our method. For instance, case $1$ for line $169$ is categorized as a normal noise since the ratio of its NPDSB to the TNSB is $0.015$, which is not even close to the proposed threshold.

\section{Conclusion and Future Work}
Developing an efficient, fast, and practical detection mechanism for real-time operations that causes minimum changes in the structure of existing EMSs is a challenging task. In this study, by using a deeper understanding of power systems, a real-time, fast, and intelligent false data detector that flags LR attacks is introduced, designed, and evaluated. We first used power systems domain insights to identify an exploitable model for the core problem of LR attacks. Then, by proving that a simple greedy algorithm is able to solve this model to optimality, the proposed detection mechanism is designed to find the most sensitive buses with respect to their impact on a target transmission asset. The results demonstrated that the greedy algorithm is pretty fast; Algorithm $1$ takes only several milliseconds to find the global solution (for each transmission asset). 

Likewise, the efficiency of the proposed method to detect some weakened random LR attacks and its ability to distinguish malicious deviations from both random Gaussian and non-Gaussian noise errors were evaluated. According to Fig. \ref{NoiseAttackSeparation169} and Fig. \ref{NoiseAttackSeparation251}, the proposed method can detect all $4000$ random attack scenarios and distinguish them from all $4000$ sets of random noise errors, which implies a $100 \%$ success rate in detecting and distinguishing these scenarios. 

Regarding the discrepancies between AC and DC modeling of power systems, we examined the accuracy of our method by comparing the physical consequences of a random attack vector, created by problem (\ref{Algorithm_3's_Alternative_1})-(\ref{Algorithm_3's_Alternative_1}), after running both AC/DC power flows. The results demonstrated negligible differences between the DC and AC power flows on the target assets, which implies our method functionality when operators use AC equations to model transmission assets' power flows. Moreover, we should mention that the color of each circle at the right-hand side of each bus (Fig. \ref{6Bus_Bestattack&A1}), which is related to the best attack vector, does not change from red to green or vice versa in the AC modeling of power systems. Hence, based on these two facts, our method still could be applied to real-world EMSs.

In this study, we developed a detection mechanism by claiming that using protection schemes are not enough to fight against cyber-attacks. However, detecting a set of false measurements is not enough to complete all security actions against cyber-attacks without any fast and appropriate corrective action. Therefore, a good direction for future study could be developing corrective actions, which are compatible with this method.

\section*{Acknowledgment}
The authors would like to thank Dr. Charlie Colbourn at Arizona State University for his kind help through this study. Likewise, I want to thank the supports from M. Ghaljehei and R. Khalilisenobari.

\ifCLASSOPTIONcaptionsoff
  \newpage
\fi

\bibliographystyle{IEEEtran}
\bibliography{IEEEabrv,Bibliography}

\def\@IEEEBIOphotowidth{1cm}    
\def\@IEEEBIOphotodepth{1cm}   
\def\@IEEEBIOhangwidth{1.2cm}    
\def\@IEEEBIOhangdepth{1.2cm}    
\begin{IEEEbiographynophoto}{Ramin Kaviani}
received the B.Sc. and M.Sc. degrees in electrical engineering from Shahid Bahonar University of Kerman, Iran. He is a Ph.D. candidate in the School of Electrical, Computer, and Energy Engineering at Arizona State University, working under the supervision of Dr. Kory W. Hedman. His research interests include power systems operations and planning, cyber-security in smart grids, and energy markets. In Summer 2020, he was a graduate intern at the New York Independent System Operator (NYISO).
\end{IEEEbiographynophoto}
\begin{IEEEbiographynophoto}{Kory W. Hedman}
specializes in three disciplines and holds six degrees: BS in Electrical Engineering and BS in Economics (University of Washington), MS in Electrical Engineering and MS in Economics (Iowa State University), and the MS and PhD degrees in Operations Research (University of California, Berkeley). Currently, Dr. Hedman is a Program Director for the Advanced Research Projects Agency-Energy (ARPA-E) US Department of Energy. Hedman led the efforts for the Grid Optimization (GO) Competition and the PERFORM Program. Before joining ARPA-E, Hedman was funded by the ARPA-E GENI and NODES programs. Dr. Hedman’s early career research led to two prestigious awards. In January 2017, US President Barack H. Obama awarded Kory W. Hedman the Presidential Early Career Award for Scientists and Engineers (PECASE), the highest honor bestowed by the United States Government on science and engineering professionals in the early stages of their independent research careers; Hedman was nominated by the US Department of Energy. Dr. Hedman has also received the IEEE Power and Energy Society Outstanding Young Engineer Award. In his first year serving on the IEEE PES Phoenix Chapter committee, Hedman expanded activities and outreach, which led to the chapter receiving the best IEEE PES chapter award (worldwide). Through the years he served on the committee, the chapter received other local and regional awards. Starting in the Fall of 2020, Dr. Hedman will be the Director of the Power Systems Engineering Research Center (PSERC). Dr. Hedman is a senior member of IEEE and a member of INFORMS.
\end{IEEEbiographynophoto}





\vfill


\end{document}